\documentclass{commat}

\title{Trace-based cryptanalysis of cyclotomic $R_{q,0}\times R_q$-PLWE for the
non-split case}

\author{\small Iv\'an Blanco-Chac\'on, Ra\'ul Dur\'an-D\'{\i}az, Rahinatou Yuh Njah Nchiwo and Beatriz Barbero-Lucas}

\authorinfo[I. Blanco-Chac\'on]{University of Alcal\'a, Spain}{ivan.blancoc@uah.es}

\authorinfo[R. Dur\'an-D\'{\i}az]{University of Alcal\'a, Spain}{raul.duran@uah.es}

\authorinfo[R. Y. Njah Nchiwo]{Aalto University School of Science, Finland}{rahinatou.njahepousenchiwo@aalto.fi}

\authorinfo[B. Barbero-Lucas]{University College, Dublin}{beatriz.barberolucas@ucdconnect.ie}

\abstract{We describe a decisional attack against a version of the PLWE
problem in which the samples are taken from a certain proper subring of large
dimension of the cyclotomic ring $\mathbb{F}_q[x]/(\Phi_{p^k}(x))$ with $k>1$ in
the case where $q\equiv 1\pmod{p}$ but $\Phi_{p^k}(x)$ is not totally split over
$\mathbb{F}_q$. Our attack uses the fact that the roots of $\Phi_{p^k}(x)$ over suitable
extensions of $\mathbb{F}_q$ have zero-trace and has overwhelming success
probability as a function of the number of input samples. An implementation in Maple and
some examples of our attack are also provided.}

\keywords{Polynomial Learning With Errors, Ring Learning With Errors, Lattice-based Cryptography}

\msc{94A60, 68W20, 12-04}

\VOLUME{31}
\NUMBER{2}
\firstpage{115}
\DOI{https://doi.org/10.46298/cm.11153}

\begin{document}

\section{Introduction}\label{s1}

One of the features which makes lattice-based cryptography so attractive is the
fact that the security of its schemes is based on worst-case versions
of classical
lattice problems, like the $\gamma$-approximate Shortest Vector problem (SVP).
If $\mathcal{S}$ is one of such schemes, breaking $\mathcal{S}$ implies that one
can solve any instance of that problem with essentially the same complexity as
that with which the scheme is broken. This
property can be rephrased by stating that the $\gamma$-SVP reduces to the scheme
$\mathcal{S}$ or that $\mathcal{S}$ admits a~reduction from the $\gamma$-SVP.

The first scheme based on the worst-case $\gamma$-SVP (for $\gamma(n)=n^c$,
fixed $c>0$, and an $n$-dimensional lattice) dates back to 1996 and is
due to Ajtai (\cite{ajtai}). This scheme
and subsequent refinements by Dwork, Cai, Nerurkar, Goldreich, Goldwasser,
Halevi and Micciancio (to cite only a~few) deal with one-way functions and
specially with the difficult issue of collision resistance. But it was not until 2005,
with Regev's pioneering work~\cite{regev}, that lattice-based methods reshaped the
landscape of public key cryptography, notably with the arising interest towards
post-quantum cryptography. Regev's scheme is based upon the so-called Learning
With Errors Problem (LWE), which roughly speaking consists in guessing a~secret
vector $\mathbf{s}\in\mathbb{F}_q^n$ if an adversary is given access to an
arbitrary number of pairs
$(\mathbf{a}_i,\langle\mathbf{a}_i,\mathbf{s}\rangle+e_i)\in\mathbb{F}
_q\times\mathbb{F}_q$ where $e_i\in\mathbb{F}_q$ are randomly sampled from a
discrete version of a~Gaussian distribution with small enough (but
not too small) variance. As Regev proves, this problems admits a~reduction from
the worst-case $\gamma$-SVP in quantum polynomial time for
$\gamma(n)=\mathcal{O}(n)$, which though it is not known to be NP-hard, is
reasonably \emph{not to far} from a~version which indeed is
proved to be so, namely, the
same problem but for $\gamma=\mathcal{O}(1)$ (see~\cite{micciancio}).

Unfortunately, Regev's cryptosystem is not practical for implementations and
deployment on average to large volumes of data, since the correction of the
scheme requires key sizes of order $\mathcal{O}(n^2)$. This drawback led
Stehl\'e to introduce the Polynomial Learning With Errors (PLWE, see~\cite{stehle}) problem and cryptosystem and later on, to Lyubashevsky, Peikert
and Regev to introduce the Ring Learning With Errors (RLWE, see~\cite{lpr})
problem, a~version of LWE where the public and secret keys are taken from a~ring
(a quotient polynomial ring in the PLWE case or a~quotient of the ring of
integers of a~number field for RLWE), rather than from the sheer vector space
$\mathbb{F}_q^n$.

Each of these two problems has its own virtues and drawbacks.
Security reduction proofs have been given for RLWE: in~\cite{lpr}, the
authors give a~reduction from the $\gamma$-SVP to the decisional version of RLWE
(for cyclotomic number fields, although in~\cite{rsw} the cyclotomic condition
is replaced by the far more general condition for the underlying number field to
be Galoisian). However, PLWE is more suitable for efficient implementations
thanks to very fast multiplication algorithms like Toom, Karatsuba or versions of the Number
Theoretic Transform (NTT) which are not available for number fields, where just
finding integral bases becomes cumbersome even for moderately large degree
and discriminant (let alone those of cryptographic size). Luckily, in a
good number of interesting cases both problems are \emph{equivalent} (see~\cite{blanco, blanco1, blanco2, dd, rsw, italians, italians2}).

The first attacks on chosen parameters for the PLWE problem are presented in~\cite{ELOS:2016:RCN} and~\cite{ELOS:2015:PWI} and they are valid if $f(x)$ has a
root $\rho\in\mathbb{F}_q$ such that either \emph{(i)} $\rho=1$, or \emph{(ii)} the
multiplicative order of $\rho$ is \emph{small}, or \emph{(iii)} the representative of
$\rho$ between $0$ and $q-1$ is also \emph{small} (say, $\rho=2$ or $3$). If
$d$ is the smallest positive integer such that $q^d \equiv 1\pmod {n}$, it is a
well-known fact that the $n$-th cyclotomic polynomial splits into
$\phi(n)/d$ irreducible degree $d$ factors over $\mathbb{F}_q[x]$, whose
roots have maximal order $n$ in the multiplicative group $\mathbb{F}_{q^d}^*$.
When $d=1$, such polynomial splits totally in $\mathbb{F}_q$ in which case it has
precisely $\phi(n)$ roots of maximal order $n$ and hence these attacks do not
apply to the cases \emph{(i)} and \emph{(ii)}.

In~\cite{blanco3}, the authors present an attack against PLWE in the case where
$f(x)$ has a quadratic irreducible factor over $\mathbb{F}_q[x]$ of the form
$x^2 +\rho$, where either $\rho=1$ or the multiplicative order of $-\rho$
determines a~\emph{smallness region} $\Sigma$ (see Section~\ref{s2} for details) such
that $|\Sigma|<q$. Apart from the order, the cardinality of $\Sigma$ depends on
the degree of the polynomial modulus as well as on the noise parameter and the
success of our attack just depends on \emph{(i)} the feasibility of constructing
$\Sigma$, \emph{(ii)} the fact that $|\Sigma|$ is upper bounded by $q$, and
\emph{(iii)} the fact
that the main loop in our algorithm can be performed with complexity
$O(\sqrt{p(p-1)(q-1)}q)$ where $n=p^k$.

The present communication deals with the case where $f(x)=\Phi_{p^k}(x)$ and
$q=1+p^2 u$ with $u$ coprime to $p$. In this case (see Section~\ref{s2}), $f(x)$
decomposes into $p(p-1)$ irreducible factors on $\mathbb{F}_q[x]$, each of degree
$p^{k-2}$ and each of these has a~root over $\mathbb{F}_{q^{p^{k-2}}}$ of trace
zero. Leaving aside the splitting case and the case where $\Phi_{p^k}(x)$
remains irreducible over $\mathbb{F}_q$, this work and~\cite{blanco3} can
somehow be considered as \emph{extreme} cases: here the irreducible factors
have maximal degree (namely, $p^{k-2}\neq 0$) whereas in~\cite{blanco3} the degree
is minimal (namely, 2).

In both works, we exploit the existence of the zero-trace root to produce a~very
effective decisional attack against a~variant of the PLWE problem, in which the
samples $(a(x),b(x))$ belong to $R_{q,0}\times R_q$, where $R_{q,0}$ is a
subring of $R_q$ that, as $\mathbb{F}_q$-subspace, has either dimension $n-1$ (in~\cite{blanco3}), or dimension $p^{n-1}(p-1)-p^{n-2}+1$ (in this work). It is the maximality
of this dimension what allowed us to reduce the PLWE to its $R_{q,0}\times
R_q$-version in probabilistic polynomial time in~\cite{blanco3} though,
unfortunately, the reduction is still unclear for the present~case.

As we discuss at the end of Section~\ref{s3}, the reduction would still be possible in
the hypothetical case that a~surjective ring homomorphism existed from $R_q$ to
$R_{q,0}$ so that small residues were taken to small residues. But the
existence of this morphism is currently uncertain and left as an open problem.

The present work is organised as follows: in Section~\ref{s2} we recall the definitions
of the RLWE and PLWE and review in a~very sketchy way (with due references
provided) how and in which sense these problems admit reductions from supposedly
hard problems dealing with ideal lattices.
We also recall several properties on the factoring of
cyclotomic polynomials in the prime power conductor case since they will be
applied in our attack. The first subsection of Section~\ref{s3} recalls the attacks for
$\theta=1$ and $\theta$ of small order in~\cite{ELOS:2016:RCN} and~\cite{ELOS:2015:PWI} and discusses their limitations in the polynomial setting. The
second subsection introduces our attack on the $R_{q,0}\times R_q$-PLWE problem
and gives a~detailed proof of its complexity. Even if the proof of the success of
our attack is essentially the same as the one given in~\cite{blanco3}, it is
repeated here to make the work self-contained. Finally, in Section~\ref{s4} we provide
numerical simulations of our algorithm in Maple and comment on its performance.

Finally, the authors are thankful to the referee for making helpful suggestions which helped to improve the quality of our manuscript.

\section{The R/PLWE problems and their relation with ideal lattices}\label{s2}

In this section we recall the definition of the Polynomial Learning With Errors
problem (PLWE) as well as we explain its relation with two supposedly
\emph{hard} problems about lattices: the shortest vector problem (SVP) and the
bounded distance decoding problem (BDD) over ideal lattices.

Though the present work aims at presenting certain types of attacks against
PLWE, we will give definitions for both RLWE and PLWE problems for the sake of
completeness and because they are intimately connected (actually equivalent in
certain suitable settings).

Let us recall first two kinds of random variables that will intervene in our
definitions:

\begin{definition}
Given an $\mathbb{F}_q$-vector space $\mathcal{V}$ of dimension
$d$, we say that a~random variable $X$ with values over $\mathcal{V}$ is uniform
if $P[X=v]=1/q^d$ for each $v\in \mathcal{V}$.
\end{definition}
In Section~\ref{s3} we will need this fact, which has a~standard proof:
\begin{lemma}
If $X_1,X_2,\dotsc,X_n$ are independent uniform distributions over
$\mathbb{F}_q$ then, for each
$\lambda_1,\lambda_2,\dotsc,\lambda_n\in\mathbb{F}_q$, not all of them zero, the
variable $\sum_{i=1}^n \lambda_i X_i$ is also uniform.
\label{unifl}
\end{lemma}
\begin{proof}
It is clearly sufficient to check that if $X_1$ and $X_2$ are uniform, then $X_1+X_2$ is also uniform. But for $i\in\mathbb{F}_q$, using the Total Probability Theorem
we have
\begin{align*}
P[X_1+X_2=i]
&=\sum_{j\in\mathbb{F}_q}P[X_1+X_2=i|X_2=j]P[X_2=j] \\
&=\frac{1}{q}\sum_{j\in\mathbb{F}_q}P[X_1=i-j]=\frac{1}{q}.
\qedhere
\end{align*}
\end{proof}

The second kind of random variable already requires to recall some notions on
lattices. Here we are following Section 2 of~\cite{lpr}.

Let $n$, $s_1$ and $s_2$ be either zero or natural numbers with $n=s_1+2s_2$. Let us consider
the $\mathbb{R}$-vector subspace of $\mathbb{C}^n$ defined as
\[
\Lambda_n=\{(x_1,\dotsc,x_n)\in\mathbb{R}^{s_1}\times\mathbb{C}^{2s_2}:
x_{s_1+i}=\overline{x}_{s_1+s_2+i}\mbox{ for }1\leq i\leq s_2\},
\]
which, endowed with the induced Hermitian metric in $\mathbb{C}^n$, is a
Euclidean space of dimension~$n$.

For $r>0$, the Gaussian function
$\rho_r(\mathbf{x})=\exp(-\pi||\mathbf{x}||^2 /r^2)$ defines, once normalised,
the density function of a~Gaussian random variable with null vector of means and
covariance matrix $r I_n$, with $I_n$ the $n$-dimensional identity matrix.
Moreover, by fixing a~basis
$\{\mathbf{h}_i\}_{i=1}^n$ of $\Lambda_n$ and for a~vector
$\mathbf{r}=(r_1,\dotsc,r_n)\in\mathbb{R}_{+}^n$ such that
$r_{s_1+i}=r_{s_1+s_2+i}$ for $1\leq i\leq s_2$, if $\{\mathcal{N}(0,
r_i)\}_{i=1}^n$ is a~set of independent $1$-dimensional zero-mean Gaussian
variables, of variance $r_i^2$, the variable
$D_{\mathbf{r}}=\sum_{i=1}^n \mathcal{N}(0,r_i)\mathbf{h}_i$ is an elliptic
$n$-dimensional zero-mean Gaussian variable whose covariance
matrix has the vector $\mathbf{r}$ as main diagonal and $0$ elsewhere. Denote by
$\rho_{\mathbf{r}}(\mathbf{x})$ the density function of $D_{\mathbf{r}}$.

For us, by a~lattice over $\Lambda_n$, we will understand a~pair
$(\mathcal{L},\iota)$ where $\mathcal{L}$ is a~finitely generated abelian group
and $\iota:\mathcal{L}\to\Lambda_n$ is a~group monomorphism. We will only deal
with full-rank lattices in this communication, i.e., those whose
$\mathbb{Z}$-rank is precisely $n$, the ambient space dimension.

If $\mathcal{L}$ is such a~full-rank lattice embedded in $\Lambda_n$ and
$\{\mathbf{h}_i\}_{i=1}^n$ is a
$\mathbb{Z}$-basis of $\iota(\mathcal{L})$ and henceforth a~basis of $\Lambda_n$
as a~vector
space, we can define the notion of a~Gaussian variable supported on
$\mathcal{L}$ as well as its discrete version, a~key ingredient for the problems
under study:

\begin{definition}
A random variable $X$ supported on $\mathcal{L}$ (hence discrete) is
called
a discrete elliptic Gaussian random variable whenever its probability function
is
\[
P[X=\mathbf{x}]=\frac{\rho_{\mathbf{r}}(\mathbf{x})}{\rho_{\mathbf{r}}(\mathcal{
L})}\mbox{ for }\mathbf{x}\in\mathcal{L}.
\]
\end{definition}
Figure 1 shows an example of $2$-dimensional discrete Gaussian, where as
expected, most of the probability mass is located around the mean vector, the
origin in this case.

\begin{figure}
  \centering
   \includegraphics[width=0.8\textwidth]{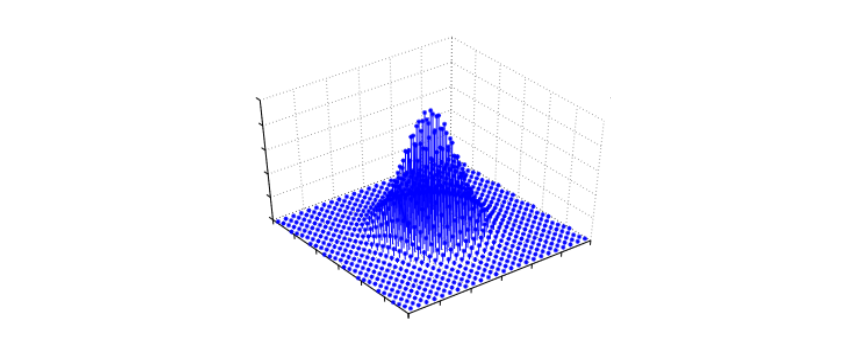}
\caption{Discrete Gaussian on $\mathbb{Z}^2$ (with permission of Oded Regev)}
\end{figure}

\def\figurename{Algorithm}

\subsection{The R/PLWE problems} Here we denote by $K/\mathbb{Q}$ a~Galois
extension of degree $n$, or equivalently, the splitting field of a~monic
irreducible polynomial $f(x)\in\mathbb{Z}[x]$ of degree $n$, minimal polynomial
of, say, $\alpha\in K$. Each automorphism of the Galois group
$\mathrm{Gal}(K/\mathbb{Q})$ is hence determined by its value at $\alpha$.
Denote these automorphisms by $\{ \sigma_i\}_{i=1}^n$ with $\sigma_1=I$, the
identity, and let us label
the roots of $f(x)$ such that $\{\alpha_i\}_{i=1}^{s_1}$ is the set of real roots and
$\{\alpha_i\}_{i=s_1+1}^n$ is the set of
$s_2$ pairs of complex non-real roots with $\alpha_i=\overline{\alpha_{i+s_2}}$
for
$s_1+1\leq i\leq s_1+s_2$. When $s_1=0$ we say that $K$ is totally complex and
when $s_2=0$ we say that $K$ is totally real.

As usual, the notation $\mathcal{O}_K$ stands for the ring of integers of $K$.
We will moreover assume, for the sake of simplicity, that $K$ is monogenic, i.e.
that $\mathcal{O}_K=\mathbb{Z}[\alpha]$ for some $\alpha\in\mathcal{O}_K$. Let
us denote $R:=\mathbb{Z}[x]/(f(x))\simeq \mathbb{Z}[\alpha]$ and for a~prime $q
\in
\mathbb{Z}$, let us set $R_q:=R/qR\cong\mathbb{F}_q[x]/(f(x))$.

Both rings $R$ and $\mathcal{O}_K$ can be endowed with a~lattice
structure over $\mathbb{R}^n$:

\begin{definition}[The coefficient embedding]
 For the ring $R$, the coefficient
embedding is
\[
\begin{array}
{ccc}
\sigma_{coef}: R & \hookrightarrow & \mathbb{R}^n \\
\sum_{i=0}^{n-1}a_i\overline{x}^i & \mapsto & (a_0,\dotsc,a_{n-1}),
\end{array}
\]
with $\overline{x}^i$ being the class of $x^i$ modulo the principal ideal
$(f(x))$.
\end{definition}

The ring $\mathcal{O}_K$, as well known, is finitely generated over
$\mathbb{Z}$ of rank $n$ and hence admits a~lattice structure too:
\begin{definition}[The canonical embedding]
 For the ring $\mathcal{O}_K$, the
canonical embedding is
\[
\begin{array}
{ccc}
\sigma_{can}: \mathcal{O}_K & \hookrightarrow & \Lambda_n\\
\alpha & \mapsto & (\sigma_1(\alpha),\dotsc,\sigma_n(\alpha)).
\end{array}
\]
\end{definition}

The lattices $\sigma_{can}(\mathcal{O}_K)$ and $\sigma_{coef}(R)$ inherit a
multiplicative structure from the product defined on their corresponding
domains, which motivates the following definition:

\begin{definition}[Ideal lattices]
 A lattice $\mathcal{L}$ is called an ideal lattice
if there exists a
ring $R$, an ideal $I\subseteq R$, and an additive group monomorphism $\sigma:
R\hookrightarrow\mathbb{R}^n$ such that $\mathcal{L}=\sigma(I)$.
\end{definition}

Let now $q\geq 2$ be a~prime. If $\chi$ is a~discrete Gaussian distribution
supported on either $\sigma_{can}(\mathcal{O}_K)$ or $\sigma_{coef}(R)$, we can
reduce component-wise its outputs modulo $q$. Such a~random variable is referred
to as a~discrete Gaussian modulo $q$.

Let $\chi_K$ be a~discrete Gaussian of $0$ mean and covariance matrix $\Sigma_K$
supported on the quotient $\mathcal{O}_K/q\mathcal{O}_K$ (or, rather, in the
$n$-dimensional torus
$\mathbb{T}_K=(K\otimes\mathbb{R})/\mathcal{O}_K$) and let $\chi_R$ a~discrete
Gaussian of $0$ mean and covariance matrix $\Sigma_{R}$, supported on $R_q$
embedded in the torus $\mathbb{T}_R=(R\otimes\mathbb{R})/R$.

\begin{definition}[RLWE/PLWE oracles]
Given $s\in \mathcal{O}_K/q\mathcal{O}_K$ (resp.
$R_q$), an RLWE oracle associated to the triple
$(\mathcal{O}_K/q\mathcal{O}_K,s,\chi_K)$ (resp. a~PLWE oracle attached
to the triple $(R_q,s, \chi_R)$) is a~probabilistic algorithm $\mathcal{A}_{s,
\chi_K}$ (resp.
$\mathcal{A}_{s, \chi_R}$) which runs as follows:
\begin{enumerate}
\item Samples an element $a \in \mathcal{O}_K/q\mathcal{O}_K$ (resp. in $R_q$)
from a~uniform distribution.
\item Samples an element $e\in \mathcal{O}_K/q\mathcal{O}_K$ from $\chi_K$
(resp. in $R_q$ from $\chi_R$).
\item Outputs the element $(a, b=as+e)$.
\end{enumerate}
\end{definition}

Setting
$\left(\mathcal{O}_K/q\mathcal{O}_K\right)^2 :=(\mathcal{O}_K/q\mathcal{O}
_K)\times (\mathcal{O}_K/q\mathcal{O}_K)$ and $R^2_q:=R_q\times R_q$, the
(decision version of) the R/P-LWE problems are defined as
follows:

\begin{definition}[RLWE/PLWE decision problems]
Let $\chi_K$ and $\chi_R$ be as before.
The R/P-LWE problem consists in deciding with non-negligible advantage, for a
set of samples of arbitrary size $(a_i, b_i)\in
\left(\mathcal{O}_K/q\mathcal{O}_K\right)^2$ (resp. in $R_q^2$), whether they
are sampled from the R/P-LWE oracle or from the uniform distribution.
 \end{definition}

From now on we will deal with the PLWE problem so that $R$ will be embedded into
$\mathbb{R}^n$ via the
coefficient embedding and so that our discrete Gaussians will be supported on
the quotient ring $R_q$.

\subsection{Lattice related problems.}
As we pointed out in the introduction,
the problems LWE, RLWE, (and PLWE whenever its equivalent to RLWE) admit quantum
polynomial time reductions from several versions of the SVP, which we recall
next:
\begin{definition}[The $\gamma$-SVP]
 Let $\gamma:\mathbb{N}\to\mathbb{R}_{+}$ be a
function. Given a~full-rank lattice $\Lambda$ of rank $n$, together with a
$\mathbb{Z}$-basis of $\Lambda$, the $\gamma$-Shortest Vector Problem consists
in returning an element $v\in\Lambda\setminus\{0\}$ such that
\[
||v||_2\leq \gamma(n)\lambda_1(\Lambda),
\]
where $\lambda_1(\Lambda)=\min_{x\in\Lambda\setminus\{0\}}||x||_2$ and
$||\phantom{a}||_2$ denotes the Euclidean norm on $\mathbb{R}^n$ (although the
problem admits, obviously, a~version with respect with any $l_p$ norm).
\end{definition}

In~\cite{micciancio}, the author proves that the $\gamma$-SVP is NP hard for
$\gamma(n)\leq\sqrt{2}$ and $n \geq 1$. As already mentioned in the
introduction, the author of~\cite{regev} proves that there exists a~quantum
polynomial time algorithm $R$ with complexity $\mathcal{O}(p(n))$ for certain
polynomial $p(x)$ that gives a~reduction from the worst-case $\gamma$-SVP with
$\gamma(n)=\mathcal{O}(n)$ to the LWE problem. This means that if an adversary
$\mathcal{A}$ existed that were able to solve LWE with non-negligible advantage
with complexity $\mathcal{O}(f(n))$, then this adversary could be turned into an
adversary $R(\mathcal{A})$ able to solve $\gamma$-SVP with complexity
$\mathcal{O}(f(n)p(n))$, also with non-negligible advantage.

As for RLWE,~\cite{lpr} shows the existence of a~quantum polynomial time
algorithm $R$ with complexity $\mathcal{O}(p(n))$ that gives a~reduction from
the worst-case $\gamma$-SVP for ideal lattices to the RLWE problem in decisional
version. The $\gamma$-SVP for ideal lattices is the restriction of
$\gamma$-SVP to the class of ideal lattices $(I,\sigma_{can})$ where $I$ is an
ideal of the ring of integers of a~cyclotomic field (see next subsection) and
$\sigma_{can}$ is the canonical embedding. In~\cite{rsw}, the authors elaborate
on the same ideas and generalise the class of ideal lattices for which the
reduction exists to those corresponding to rings of integers of Galois number
fields.

It is convenient to point out that the NP-hardness of SVP is established to
uniformly bounded functions $\gamma$, hence it is not clear that LWE is NP-hard,
even if empirical evidence strongly suggests that it is an \emph{intractable}
problem. Unfortunately, for RLWE the situation is even weaker, since it is not
even known whether $\gamma$-SVP is NP-hard for ideal lattices for uniformly
bounded $\gamma$, which is currently an active research area.

\subsection{The cyclotomic polynomial and its splitting behaviour over finite
fields}

We will denote $K_n:=\mathbb{Q}(\zeta_n)$, the $n$-th
cyclotomic field (where $\zeta_n$ denotes a~primitive complex $n$-root of
unity). It is well known that $K_n$ is the splitting field of the $n$-th
cyclotomic
polynomial, which we will denote by $\Phi_n(x)$. In particular, $K_n/\mathbb{Q}$
is a~Galois extension of degree $m:=\phi(n)$, where $\phi$ stands for the
Euler's totient function. It is also well known that $K_n$ is monogenic, in
particular $\mathcal{O}_K=\mathbb{Z}[\zeta_n]$.

When $q\equiv 1\pmod{n}$, the prime $q$ is totally split in $\mathcal{O}_K$ and
hence $\Phi_n(x)$ has $m$ different roots in
$\mathbb{F}_q$, all of them of maximal multiplicative order $n$. We will deal,
however, with the non-totally split case and moreover, we will suppose that
$n=p^k$ for a~prime $p$. The following result addresses the factorisation of
$\Phi_n(x)$ into irreducible factors in $\mathbb{F}_q[x]$:

\begin{theorem}[\cite{WZFY:2017:EFC}]
 Let $q=1+p^A u$, with $A\geq 1$, and $p$, $q$
primes. Suppose that $(u,p)=1$ and denote by $\Omega(p^A)$ the group of
primitive $p^A$-th roots of unity in $\mathbb{F}_q$. Assume $n> A$. Then, we
have:
\[
\Phi_{p^n}(x)=\prod_{\rho\in\Omega(p^A)}\left(x^{p^{n-A}}-\rho\right),
\]
where the polynomials $x^{p^{n-A}}-\rho$ are irreducible over $\mathbb{F}_q$.
\label{chinese}
\end{theorem}

We have the following straightforward consequence which will be useful later on:

\begin{corollary}
Notations as in Theorem~\ref{chinese}, for every $v\in\mathbb{N}$
such
that $(v,p)=1$, for each $\rho\in\Omega(p^A)$ and for each $0\leq k<n-A$, the
polynomial $x^{p^{n-k-A}}-\rho^v$ is irreducible over $\mathbb{F}_q[x]$.
\label{lemairr}
\end{corollary}
\begin{proof}
First, notice that $\rho^v$ is also a~primitive $p^A$-th root of
unity, hence if we could express $x^{p^{n-k-A}}-\rho^v =f(x)g(x)$ with
$deg(f(x)),deg(g(x))\geq 1$, it would follow that
$x^{p^{n-A}}-\rho^v =f(x^{p^k})g(x^{p^k})$, a~contradiction.
\end{proof}

\subsection{Fast evaluation of polynomials over finite fields} One of the issues
we must confront is to evaluate polynomial expressions over a~finite field at
elements of certain extensions, in the most efficient possible manner. In
particular, for cyclotomic prime conductors of almost-cryptographic size, the
well-known Horner's algorithm might easily become inefficient. The method that
we will use, due to Elia, Rosenthal and Schipani, is called \emph{automorphic
evaluation} drastically reduces the number of $\mathbb{F}_q$-products by
a square root factor.

\begin{theorem}[{\cite[Theorem 3]{ERS:2012:PEF}}]
 The minimum number of
$\mathbb{F}_q$-products
required to evaluate a~polynomial of degree $n$ with coefficients in
$\mathbb{F}_{q^s}$ at an element of $\mathbb{F}_{q^m}$ with $m\geq s$, is upper
bounded by
\[
2s(\sqrt{n(q-1)}+ 1/2).
\]
\label{auto}
\end{theorem}

\section{\texorpdfstring{An attack based on traces over finite extensions of
$\mathbb{F}_q$}{An attack based on traces over finite extensions of
Fq}}\label{s3}

One of the first attacks on PLWE (and on RLWE whenever they are equivalent), is
that described in~\cite{ELOS:2016:RCN} and~\cite{ELOS:2015:PWI}, which for a
quotient ring $R_q=\mathbb{F}_q[x]/(f(x))$ and a~prime $q$, are applicable and
successful whenever it exists a~simple root $\alpha\in\mathbb{F}_q$ such that
\begin{enumerate}
\item[a)] $\alpha=1$, or
\item[b)] $\alpha$ has \emph{small order} modulo $q$, or
\item[c)] $\alpha$ has \emph{small residue} modulo $q$.
\end{enumerate}
By \emph{small order} the authors understand orders of up to $5$, while by
\emph{small residue}, they mean $\alpha=2$ or $\alpha=3$ (modulo $q$).

Using the Chinese remainder theorem, we express
\[
R_q\simeq \mathbb{F}_q[x]/(x-\alpha)\times\mathbb{F}_q[x]/(h(x)),
\]
where $h(x)$ is coprime to $x-\alpha$. We obtain hence a~ring homomorphism
\[
\psi_\alpha\colon R_q\rightarrow\mathbb{F}_q[x]/(x-\alpha)\simeq\mathbb{F}_q,
\]
which is nothing else than the evaluation-at-$\alpha$ map, namely,
$\psi_\alpha(g(x)) = g(\alpha)$,
for each $g(x)\in R_q$.

Let $n$ be the degree of $f(x)$. Next, we describe two of the three attacks
presented in~\cite{ELOS:2016:RCN}, namely, those corresponding to the cases a)
and b) above. Let us suppose, to start with, that $\alpha=1$ is a~root of
$f(x)$. For a
PLWE sample $(a(x),b(x)=a(x)s(x)+e(x))$, the error term,
$e(x)=\sum_{i=0}^{n-1}e_ix^i$, has its coefficients
$e_i\in\mathbb{F}_q$ sampled from a~discrete Gaussian with small
enough standard deviation $\sigma$ (the authors set $\sigma\cong 8$, as
suggested by applications). For an element $s\in\mathbb{F}_q$, writing $s=s(1)$
and
applying the evaluation map, we have
\[
b(1)-a(1)s=e(1)=\sum_{i=0}^{n-1}e_i,
\]
and the sum $\sum_{i=0}^{n-1}e_i$ is hence sampled from a~discrete
Gaussian variable of standard deviation $\sqrt{n}\sigma$, which according to
practical specifications, is of order $O(q^{1/4})$.

For a~right guess $s=s(1)$, the value $b(1)-a(1)s$ will belong to the set of
integers $[-2\sqrt{n}\sigma, 2\sqrt{n}\sigma]\cap\mathbb{Z}$ (which can be
easily enumerated) with probability about $0.95$. Hence, we will refer to the
set $[-2\sqrt{n}\sigma, 2\sqrt{n}\sigma]\cap\mathbb{Z}$ as the \emph{smallness
region} for this attack.

Case b) is more subtle. Indeed, let $r$ denote the multiplicative order of a
root $\alpha\in\mathbb{F}_q$, such that $\alpha\ne 1$. Given a~PLWE sample
$(a(x),b(x)=a(x)s(x)+e(x))$, for $s=s(\alpha)\in\mathbb{F}_q$, we have
\[
b(\alpha)-a(\alpha)s=e(\alpha) = \sum_{j=0}^{r-1}\sum_{i=0}^{\frac{n}{r}-1}
e_{ir+j}\alpha^j,
\]
assuming without loss of generality that $r\mid n$.

The elements $e_j = \sum_{i=0}^{n/r-1}e_{ir+j}$ can be regarded as sampled from
a Gaussian
distribution of $0$ mean and standard deviation $\sqrt{n/r}\sigma$, and thus
they
belong to the set of integers $[-2\sqrt{n/r}\sigma,
2\sqrt{n/r}\sigma]\cap\mathbb{Z}$ with probability $0.95$. This leads us to
consider in this case the \emph{smallness region} as the set $\Sigma$ of all
possible values for $e(\alpha)$, which can be precomputed and stored in a
look-up table. Notice that
\[
|\Sigma| \leq \left(4\sqrt{n/r}\sigma + 1\right)^r.
\]

\begin{figure}[ht]
\centering
\begin{tabular}[c]
{ll}
\hline
\textbf{Input:} & A collection of samples $C=\{(a_i(x),b_i(x))\}
_{i=1}^M \subseteq R_q^2$ \\
 & A look-up table $\Sigma$ of all possible values for $e(\alpha)$ \\
\textbf{Output:} & A guess $g\in\mathbb{F}_q$ for $s(\alpha)$,\\
 & or \textbf{NOT PLWE},\\
 & or \textbf{NOT ENOUGH SAMPLES}\\
\hline
\end{tabular}

\begin{itemize}
\item \texttt{\emph{set}} $S:=\mathbb{F}_q$
\item \texttt{\emph{set}} $G:=\emptyset$
\item \texttt{\emph{for}} $g\in S$ \texttt{\emph{do}}
	
\begin{itemize}
	\item \texttt{\emph{for}} $(a_i(x),b_i(x))\in C$ \texttt{\emph{do}}
		
\begin{itemize}
		\item \texttt{\emph{if}} $b_i(\alpha)-a_i(\alpha)g\notin \Sigma$
\texttt{\emph{then}}
			
\begin{itemize}
			\item \texttt{\emph{next}} $g$
			
\end{itemize}
		
\end{itemize}
	\item \texttt{\emph{set}} $G:=G\cup\{g\}$
	
\end{itemize}
\item \texttt{\emph{if}} $G=\emptyset$ \texttt{\emph{then return}} \textbf{NOT
PLWE}
\item \texttt{\emph{if}} $G=\{g\}$ \texttt{\emph{then return}} $g$
\item \texttt{\emph{if}} $|G|>1$ \texttt{\emph{then return}} \textbf{NOT ENOUGH
SAMPLES}
\end{itemize}
\hrule
\caption{Algorithm solving PLWE decision problem}
\label{alg1}
\end{figure}
In~\cite{ELOS:2015:PWI}, these ideas are presented and converted into
Algorithm~\ref{alg1}, whose
probability of success is also derived therein:

\begin{proposition}[{\cite[Proposition~3.1]{ELOS:2015:PWI}}]
 Assume $|\Sigma|<q$. If
Algorithm~\ref{alg1}
returns \textbf{NOT PLWE}, then the samples come from the uniform distribution.
If it outputs anything other than \textbf{NOT PLWE}, then the samples are valid
PLWE samples with probability given by $1-\left(|\Sigma|/q\right)^M$. In particular,
this probability tends to 1 as $M$ grows.
\label{pralg1}
\end{proposition}

\begin{remark}
Notice that the cyclotomic polynomial $\Phi_n(x)$ is protected
against these attacks. Indeed, $\alpha=1$ is never a~root modulo $q\neq p$.
Moreover, for $q\equiv 1\pmod{n}$ the order of each of the $m$ different roots
of $\Phi_n(x)$ is precisely $n$.
\end{remark}

\subsection{Our method. Preliminary facts}\label{ss1} In this section we present an attack
against a~variant of the PLWE problem for $\Phi_m(x)$ and for non totally-split
primes $q$ by using roots
of $\Phi_m(x)$ over finite degree extensions of $\mathbb{F}_q$. To brief
notation, let $m:=p^n$, $N=\phi(m)$ and $R_q:=\mathbb{F}_q[x]/(\Phi_m(x))$. We
will assume, as in Theorem~\ref{chinese},
that $q=1+p^A u$, with $A\geq 1$, $(u,p)=1$ and that $n>A$.

Our attack starts with a~primitive $p^A$-th root $\rho$ of unity modulo $q$, for
which we take $\alpha\in\mathbb{F}_{q^{p^{n-A}}}\setminus\mathbb{F}_q$, a
$p^{n-A}$-th root of $\rho$. Due to Theorem~\ref{chinese} we have
$Tr(\alpha)=0$,
where $Tr$ stands for the trace of $\mathbb{F}_{q^{p^{n-A}}}$ over
$\mathbb{F}_q$.

Now, if $(a(x),b(x)=a(x)s(x)+e(x))\in R_q^2$ is a~PLWE sample attached to a
secret $s(x)$ and an error term $e(x)=\sum_{i=0}^{N-1}e_ix^i$,
then
\[
b(\alpha)-a(\alpha)s=e(\alpha),
\]
with $s:=s(\alpha)\in \mathbb{F}_{q^{p^{n-A}}}$ and
\begin{equation}
Tr(b(\alpha)-a(\alpha)s)=Tr(e(\alpha))=\sum_{i=0}^{N-1}e_it_i,
\label{partida}
\end{equation}
where $t_i=Tr(\alpha^i)$.

If $(i,p)=1$ then $t_i=0$ since $\alpha$ is a~root of $x^{p^{n-A}}-\rho$ and
$ord(\alpha^i)=ord(\alpha)=m$. More in general, we will make use of the
following
\begin{lemma}
Notations as before, for $i=p^k v$ with $(v,p)=1$ and $0\leq k< n-A$,
then $t_i=0$.
\label{lemtr}
\end{lemma}
\begin{proof}
For $i=p^k v$ with $(v,p)=1$ and $0\leq k< n-A$, the element $\alpha^{p^k v}$ is a
root of the polynomial $x^{p^{n-k-A}}-\rho^v$ and since $\rho^v$ is also a
primitive $p^A$-th root of unity, this polynomial is irreducible according to
Corollary~\ref{lemairr}. Hence $Tr(\alpha^{p^k v})=0$.
\end{proof}
Applying Lemma~\ref{lemtr} to the right hand side of Equation~\ref{partida}
we are
left with
\begin{equation}
Tr(b(\alpha)-a(\alpha)s)=p^{n-A}\sum_{j=0}^{p^{A-1}(p-1)-1}e_{jp^{n-A}}\rho^j.
\label{traceerror}
\end{equation}
But, again, the coefficients $e_{jp^{n-A}}$ are sampled from a~discrete Gaussian
$N\left(0, \sigma^2 \right)$ and we can list those elements which occur with
probability beyond $0.95$, namely, the integer values in the interval
$[-2\sigma, 2\sigma]$.

From now on we will suppose that $A=2$ and $\sigma=8$ so that in $[-2\sigma,
2\sigma]$ there are $32$ integers. We can construct a~look-up table where the
expression~\ref{traceerror} takes on values with large probability, namely, the
\emph{smallness region}, $\Sigma$. Observe that
\begin{equation}
|\Sigma| \leq \left(4\sigma + 1\right)^{p(p-1)}.
\label{lookup}
\end{equation}
To construct $\Sigma$ requires $32^{p(p-1)}$ multiplications in $\mathbb{F}_q$,
which is feasible for not very large values of $p$.

\subsection{The trace map} In order to compute the trace of an element
$\theta\in\mathbb{F}_{q^{p^{n-2}}}$, we can proceed by fixing an
$\mathbb{F}_q$-basis of $\mathbb{F}_{q^{p^{n-2}}}$. For instance, we will
stick to the power-basis $\{1,\alpha,\dotsc,\alpha^{p^{n-2}-1}\}$. Now we identify
$\mathbb{F}_{q^{p^{n-2}}}\cong \mathbb{F}_{q}^{p^{n-2}}$ and we can write
\[
\theta=\sum_{i=0}^{p^{n-2}-1}a_i\alpha^i,
\]
where $a_i\in\mathbb{F}_q$ and $\alpha^{p^{n-2}}=\rho$, our chosen $p^2$-th root
of unity in $\mathbb{F}_q$. Taking trace, which is an $\mathbb{F}_q$-linear map,
we have, as explained in Subsection~\ref{ss1}: 
\begin{equation}
Tr(\theta)=\sum_{i=0}^{p^{n-2}-1}a_iTr(\alpha^i)=p^{n-2}a_0.
\label{eqtrace}
\end{equation}
A first tentative approach to exploit a~root $\alpha\in\mathbb{F}_{q^{p^{n-2}}}$
for an attack would be to run over the elements $s$ in this field as putative
guesses for $s(\alpha)$ and to decide for each sample $(a(x),b(x))\in R_q^2$
whether $Tr(b(\alpha)-a(\alpha)s)$ belongs or not to the \emph{smallness region
}$\Sigma$. One would need to evaluate $Tr(b(\alpha))$ and $Tr(a(\alpha)s)$ for
each $s\in\mathbb{F}_q^{p^{n-2}}$. Leaving aside that running through all the
elements of this large field is definitely unfeasible, we can evaluate
$Tr(b(\alpha))$, which is independent of $s$, by applying Lemma~\ref{lemtr}:
\[
Tr(b(\alpha))=p^{n-2}\sum_{j=0}^{p(p-1)-1}b_{jp^{n-2}}\rho^j.
\]
Hence, evaluating $Tr(b(\alpha))$ takes about $2\sqrt{p(p-1)(q-1)}$
$\mathbb{F}_q$-products.

As for $Tr(a(\alpha)s)$, notice that the map $s\mapsto
T_{a(\alpha)}(s):=Tr(a(\alpha)s)$ is also $\mathbb{F}_q$-linear, hence
identifying $s\in\mathbb{F}_{q^{p^{n-2}}}$ with its coordinates
$(s_0,s_1,\dotsc,s_{p^{n-2}-1})$, we can write:
\begin{equation}
T_{a(\alpha)}(s)=Tr\left(\sum_{i=0}^{N-1}a_i\alpha^i \sum_{j=0}^{p^{n-2}-1}
s_j\alpha^j \right).
\label{midtrace1}
\end{equation}
Since $0\leq i\leq N-1$ and $0\leq j\leq p^{n-2}-1$, the terms for which the
trace do not vanish are those of the form $a_ix_j\alpha^{i+j}$ with
$i+j=vp^{n-2}$, with $0\leq v\leq p(p-1)$. Namely:
\begin{equation}
T_{a(\alpha)}(s)=p^{n-2}a_0s_0+p^{n-2}\sum_{v=1}^{p(p-1)}\left(\sum_{j=0}^{p^{
n-2}-1}s_ja_{vp^{n-2}-j}\right)\rho^v.
\label{astrace}
\end{equation}
Since for each $0\leq j\leq p^{n-2}-1$ we have to evaluate a~polynomial of
degree $p(p-1)$ over $\mathbb{F}_q$, which takes about $2\sqrt{p(p-1)(q-1)}$,
evaluating $Tr(a(\alpha)s)$ takes $2p^{n-2}\sqrt{p(p-1)(q-1)}$ per sample.
However, as we can see, the expression for the trace in Equation~\ref{astrace}
is
rather complicated and computationally far from optimal, specially if we have to
perform it for each sample and for each guess. For this reason, our attack is
restricted to samples $(a(x),b(x))\in R_q^2$ whose left component belong to a
subring $R_{q,0}$ which has large dimension.

\subsection{A distinguished subspace} Instead of in $R_q^2$, we consider samples
in $R_{q,0}\times R_q$ where
\[
R_{q,0}=\{p(x)\in R_q: p(\alpha)\in\mathbb{F}_q\}.
\]
\begin{proposition}
The set $R_{q,0}$ is a~subring of $R_q$ and an $\mathbb{F}_q$-vector
subspace of $R_q$ of
dimension $p^{n-1}(p-1)-p^{n-2}+1$.
\label{dim0}
\end{proposition}
\begin{proof}
It is obvious that $R_{q,0}$ is an $\mathbb{F}_q$-vector subspace
and a~subring of $R_q$. As for the dimension, notice that for
$p(x)=\sum_{i=0}^{N-1}p_ix^i$, we have, by dividing each index $i$
by $p^{n-2}$:
\[
p(\alpha)=\sum_{j=0}^{p^{n-2}-1}\left(\sum_{v=0}^{p(p-1)-1}p_{vp^{n-2}+j}
\rho^v \right)\alpha^j,
\]
hence $p(\alpha)\in\mathbb{F}_q$ if and only if
$\sum_{v=0}^{p(p-1)-1}x_{vp^{n-2}+j}\rho^v =0$ for each $0< j\leq
p^{n-2}-1$. These are $p^{n-2}-1$ linearly independent equations, hence the
result follows.
\end{proof}
\begin{remark}
Observe that, for $a(x)\in R_{q,0}$, it holds that
\[
Tr(a(\alpha)s)
=a(\alpha)Tr(s)
=p^{n-2}a(\alpha)s_0,
\]
which requires only two
$\mathbb{F}_q$-multiplications to compute.
\end{remark}

\subsection{An attack on $R_{q,0} \times R_q$-PLWE} Denote
$S:=\mathbb{F}_{q^{p^{n-2}}}$
and assume that we are given a~set of samples from $R_{q,0}\times R_q$.
The goal is to distinguish whether these samples come from the $R_{q,0}$-PLWE
distribution or from a~uniform distribution with values in $R_{q,0}\times R_q$.
To that end, given a~sample $(a_i(x),b_i(x))$, we pick a~guess $s\in S$ for
$s(\alpha)$ and check whether
$e_{i}:=\frac{1}{p^{n-2}}Tr(b_i(\alpha)-a_i(\alpha)s)$ belongs to the
look-up table $\Sigma$. 
If this is not the
case, we
can safely remove from $S$ not only $s$, but also all the elements $t\in
\mathbb{F}_{q^{p^{n-2}}}$ with the same trace as $s$. But notice that if
$s=\sum_{j=0}^{p^{n-2}-1}s_j\alpha^j$, then an element
$t=\sum_{j=0}^{p^{n-2}-1}t_j\alpha^j$ has the same trace as $s$
if and only if $t_0=s_0$. Hence, given an $s\in S$, if we find a
sample $(a_i(x),b_i(x))$ for which $e_i\notin\Sigma$, then we can delete
$q^{p^{n-2}-1}$ elements of $S$.

Since $a(\alpha)\in\mathbb{F}_q$, then
\[
\frac{1}{p^{n-2}}Tr(b(\alpha)-a(\alpha)s)=\frac{1}{p^{n-2}}Tr(b(\alpha))-\frac{1
}{p^{n-2}}a(\alpha)Tr(s).
\]
Therefore, it is enough just to check, for each $g\in\mathbb{F}_q$ (so that
$g$ is a~putative value for $Tr(s)$), whether or not
\[
\frac{1}{p^{n-2}}Tr(b(\alpha))-\frac{1}{p^{n-2}}a(\alpha)g\in\Sigma.
\]
This is at the price that if the algorithm returns just an element
$g\in\mathbb{F}_q$, we should understand that this element is just the trace of
one of the $q^{p^{n-2}-1}$ possible guesses for $s(\alpha)$. However, this is
(even if weaker than Algorithm~\ref{alg1}) enough as a~decision attack.

\begin{figure}[ht]
\centering
\begin{tabular}[c]
{ll}
\hline
\textbf{Input:} & A set of samples $C=\{(a_i(x),b_i(x))\}_{i=1}^M \in
R_{q,0}\times R_q$ \\
 & A look-up table $\Sigma$ of all possible values for $Tr(e(\alpha))$ \\
\textbf{Output:} & \textbf{PLWE},\\
                 & or \textbf{NOT PLWE},\\
                 & or \textbf{NOT ENOUGH SAMPLES}\\
\hline
\end{tabular}

\begin{itemize}
\item \texttt{\emph{set}} $G:=\emptyset$
\item \texttt{\emph{for}} $g\in \mathbb{F}_q$ \texttt{\emph{do}}
	
\begin{itemize}
	\item \texttt{\emph{for}} $(a_i(x),b_i(x))\in C$ \texttt{\emph{do}}
		
\begin{itemize}
		\item \texttt{\emph{if}} $\frac{1}{p^{n-2}}\left(Tr(b(\alpha))
                                -a(\alpha)g\right)\notin\Sigma$
\texttt{\emph{then}}
			
\begin{itemize}
			\item \texttt{\emph{next}} $g$
			
\end{itemize}
		
\end{itemize}
	\item \texttt{\emph{set}} $G:=G\cup\{g\}$
	
\end{itemize}
\item \texttt{\emph{if}} $G=\emptyset$ \texttt{\emph{then return}} \textbf{NOT
PLWE}
\item \texttt{\emph{if}} $|G|=1$ \texttt{\emph{then return}} \textbf{PLWE}
\item \texttt{\emph{if}} $|G|>1$ \texttt{\emph{then return}} \textbf{NOT ENOUGH
SAMPLES}
\end{itemize}
\hrule
\caption{Decision attack against $R_{q,0}$-PLWE}
\label{alg2}
\end{figure}

Observe that if $|G|=1$, say $G=\{g\}$, unlike
Algorithm~\ref{alg1}, our Algorithm~\ref{alg2} does not output a~guess for
$s(\alpha)$;
all we can only suspect is that there likely exists $\tilde{s}\in
\mathbb{F}_{q^{n-2}}$ such that $Tr(\tilde{s})=g$ with $s(\alpha)=\tilde{s}$.

Next, we evaluate the complexity of our attack in terms of
$\mathbb{F}_q$-multiplications:
\begin{proposition}
Given $M$ samples in $R_{q,0}\times R_q$, the number of
$\mathbb{F}_q$-multiplications required for Algorithm~\ref{alg2} is, at worst,
of
order $\mathcal{O}(\sqrt{p(p-1)(q-1)}Mq)$.
\end{proposition}
\begin{proof}
To begin with, given $g\in \mathbb{F}_q$:
\begin{itemize}
\item For each sample $(a_i(x),b_i(x))$, evaluating $Tr(b_i(\alpha))$, by
automorphic evaluation requires $2\sqrt{p(p-1)(q-1)}$ multiplications in
$\mathbb{F}_q$. Therefore checking whether the element
$\frac{1}{p^{n-2}} Tr(b(\alpha)-a(\alpha)g)$ is in $\Sigma$ requires
$2\sqrt{p(p-1)(q-1)}+2$ multiplications in~$\mathbb{F}_q$.
\item In the worst case, the condition will fail for all the samples, in
which case we will perform $(2\sqrt{p(p-1)(q-1)}+2)M$ multiplications in
$\mathbb{F}_q$ for each $g\in\mathbb{F}_q$.
\end{itemize}
Since the previous steps must be performed for every $g\in\mathbb{F}_q$, the
number of multiplications for the worst case will be
$(2\sqrt{p(p-1)(q-1)}+2)Mq$.
\end{proof}

To derive the success probability of our attack we will make use of the
following:
\begin{remark}
Given an input sample $(a(x),b(x))\in R_{q,0}\times R_q$ for
Algorithm~\ref{alg2}, given
$g\in\mathbb{F}_q$ such that $a(x)=\sum_{j=0}^{p^{n-1}(p-1)-1}a_jx^j$ and
$b(x)=\sum_{j=0}^{p^{n-1}(p-1)-1}b_jx^j$, one can notice that
checking whether
$\frac{1}{p^{n-2}}Tr(b(\alpha))-\frac{1}{p^{n-2}}Tr(a(\alpha))g\in\Sigma$ is
exactly the same as checking whether $b^\prime(\rho)-ga^\prime(\rho)\in \Sigma$
where
$a^\prime(x)=\sum_{j=0}^{p(p-1)-1}a_{jp^{n-2}}x^j$ and
$b^\prime(x)=\sum_{j=0}^{p(p-1)-1}b_{jp^{n-2}}x^j$, with
$a^\prime(x),b^\prime(x)\in R_q^\prime=\mathbb{F}_q[x]/(\Phi_{p^2}(x))$. Thus,
the result of
Algorithm~\ref{alg2} on samples $(a_i(x),b_i(x))\in R_q^2$ is exactly the result
of
Algorithm~\ref{alg1} applied to the samples $(a_i^\prime(x),b_i^\prime(x))$ in
$(R_q^{\prime})^2$.
\label{remark1}
\end{remark}
The following result will also be useful in our proof:
\begin{lemma}
Let $\{(a_i(x),b_i(x))\}_{i=1}^M$ be a~set of input samples for
Algorithm~\ref{alg2}, where, as usual, the $a_i(x)$ are taken uniformly from
$R_{q,0}$ with probability $q^{-d}$. Then, for the corresponding input samples
$(a_i^\prime(x),b_i^\prime(x))$ for Algorithm~\ref{alg1}, the elements
$a_i^\prime(x)$ are taken uniformly from $R_{q}^\prime$ with probability
$q^{-p(p-1)}$.
\label{unifprima}
\end{lemma}
\begin{proof}
For every sample $a(x)$ taken uniformly from $R_{q,0}$, if we
write, as in Proposition~\ref{dim0}
\[
a(x)=\sum_{j=0}^{p^{n-2}-1}\left(\sum_{v=0}^{p(p-1)-1}a_{vp^{n-2}+j}x^{vp^{n-2}}
\right)x^j,
\]
we observe that the polynomial
$a_0(x)=\sum_{v=0}^{p(p-1)-1}a_{vp^{n-2}}x^{vp^{n-2}}$ will be sampled from
$R_{q,0}$ with probability $q^{-d}$ where $d=p^{n-1}(p-1)-p^{n-2}+1$. But for
each $j\in\{1,\dotsc,p^{n-2}-1\}$ and for each $p(p-1)$-tuple
\[
(a_{j},a_{p^{n-2}+j},\dotsc,a_{(p(p-1)-1)p^{n-2}+j})
\]
such that
$\sum_{v=0}^{p(p-1)-1}a_{vp^{n-2}+j}\rho^v =0$, the polynomial
\[
a_0(x)+\sum_{j=1}^{p^{n-2}-1}\sum_{v=0}^{p(p-1)-1}a_{vp^{n-2}+j}x^{vp^{n-2}+j}
\]
is also sampled with probability $q^{-d}$. These tuples form a~vector space of
dimension $p(p-1)-1$, hence, there are $q^{p(p-1)-1}$ of such tuples for every
$j$. Hence, for Algorithm~\ref{alg1}, the input sample
$a^\prime(x)=\sum_{v=0}^{p(p-1)-1}a_{vp^{n-2}}x^v$ (notice that
$a_0(x)=a^\prime(x^{p^{n-2}})$) will occur with probability $q^{-d}P$, where
$P$ is
the number of joint samples for all the $j's$ together, namely
$P=q^{(p(p-1)-1)(p^{n-2}-1)}$, hence, the sample $a^\prime(x)$ for
Algorithm~\ref{alg1} will occur with probability
\[
q^{-d+(p(p-1)-1)(p^{n-2}-1)}=q^{-p(p-1)}.
\qedhere
\]
\end{proof}

We can now study the probability of success of our attack:

\begin{proposition}
Assume that $|\Sigma|<q$. If Algorithm~\ref{alg2} returns
\textbf{NOT
PLWE}, then the samples come from the uniform distribution on $R_{q,0}\times
R_q$. If it outputs anything else than \textbf{NOT PLWE}, then the samples are
valid PLWE samples with probability $1-\left(|\Sigma|/q\right)^M$. In
particular, this probability tends to 1 as $M$ grows.
\end{proposition}
\begin{proof}
Set input samples $S=\{(a_i(x),b_i(x))\}_{i=1}^M$ and $S^\prime=\{(a_i'(x),
b_i'(x))\}_{i=1}^M$, and let us define the following events:
\begin{itemize}
\item $E_q=$ The input samples $S$ for Algorithm~\ref{alg2} are uniform,
\item $E_q^\prime=$ The input samples $S^\prime$ for Algorithm~\ref{alg1} are
uniform,
\item $rP=$ \emph{Algorithm~\ref{alg2} returns \textbf{PLWE}} on input samples
$S$,
\item $rP^\prime=$ \emph{Algorithm~\ref{alg1} returns \textbf{PLWE}} on input
samples $S^\prime$,
\item $rNE=$ \emph{Algorithm~\ref{alg2} returns \textbf{NOT ENOUGH SAMPLES}} on
input samples $S$,
\item $rNE^\prime=$ \emph{Algorithm~\ref{alg1} returns \textbf{NOT ENOUGH
SAMPLES}} on input samples $S^\prime$,
\item $rNP=$ \emph{Algorithm~\ref{alg2} returns \textbf{NOT PLWE}} on input
samples $S$,
\item $rNP^\prime=$ \emph{Algorithm~\ref{alg1} returns \textbf{NOT PLWE}} on
input samples $S^\prime$.
\end{itemize}

We clearly have $rP\cup rNE\subseteq rP^\prime\cup rNE^\prime$. On the other
hand, if $rP^\prime\cup rNE^\prime$ holds, it is because the set $G$ of guesses
for $s'(\rho)$ in Algorithm~\ref{alg1} on input samples $S^\prime$ has at
least one element, hence, this element will also be a~guess for $Tr(s(\alpha))$
in Algorithm~\ref{alg2} on input samples $S$ and hence $rP\cup rNE$ will also
hold. Henceforth
\[
rP\cup rNE=rP^\prime\cup rNE^\prime\mbox{ and }rNP=rNP^\prime.
\]
On the other hand, as we have pointed our in Remark~\ref{remark1}, we have that
$E_q\subseteq E_q^\prime$.

Further, if $E_q^\prime$ holds then, given $s\in\mathbb{F}_q$, by using
Lemma~\ref{unifl}, the elements $b_i^\prime(\rho)-sa_i^\prime(\rho)$ are
uniformly taken on $\mathbb{F}_q$. This fact implies that the input samples for
Algorithm~\ref{alg2} cannot come from the PLWE distribution:
Otherwise, if $(a_i(x),b(x)=a_i(x)s(x)+e_i(x))$ is a~PLWE sample for
Algorithm~\ref{alg2}, with $e_i(x)=\sum_{j=0}^{p^{n-1}(p-1)-1}\! e_{ij}x^j$,
then the terms $e_{ij}$ are taken from an $\mathbb{F}_q$-valued Gaussian
$N(0,\sigma)$ and so are taken, in particular, those of the form $e_{jp^{n-2}}$.
Hence for $s=Tr(s(\alpha))$ we have
\[
\frac{1}{p^{n-2}}Tr(b(\alpha))-\frac{1}{p^{n-2}}
Tr(a(\alpha))s=b^\prime(\rho)-sa^\prime(\rho)=\sum_{j=0}^{p(p-1)-1}e_{jp^{n-2}}
\rho^j,
\]
which is a~contradiction. Hence the input samples $S$ for Algorithm~\ref{alg2}
should be uniform and $E_q=E_q^\prime$.

Hence
\[
E_q\cap(rP\cup rNE)=E_q^\prime\cap(rP^\prime\cup rNE^\prime)\mbox{ and }E_q\cap
rNP=E_q^\prime\cap rNP^\prime.
\]
Hence if the algorithm returns \textbf{NOT PLWE} then
\[
P[E_q|rNP]=\frac{P[E_q\cap rNP]}{P[rNP]}=\frac{P[E_q^\prime\cap
rNP^\prime]}{P[rNP^\prime]}=P[E_q^\prime|rNP^\prime].
\]
On the other hand, if the algorithm returns anything else than \textbf{NOT
PLWE}, then:
\begin{align}
P[E_q|rP\cup rNE]
&= \frac{P[E_q\cap(rP\cup rNE)]}{P[rP\cup rNE]} \\
&= \frac{P[E_q^\prime\cap(rP^\prime\cup rNE^\prime)]}{P[rP^\prime\cup rNE^\prime]} \\
&= P[E_q^\prime|rP^\prime\cup rNE^\prime],
\end{align}
which equals $\left(|\Sigma|/q\right)^M$ due to Proposition~\ref{pralg1}.
\end{proof}

As pointed out in the introduction, suppose that there existed a~surjective ring
homomorphism $t:R_q\to R_{q,0}$ such that
$t\circ\mathcal{N}(0,\sigma)=\mathcal{N}^* (0,\sigma^*)$ where
$\mathcal{N}(0,\sigma)$ is a~centred discrete $R_q$-valued Gaussian
distribution of parameter $\sigma$ and $\mathcal{N}^* (0,\sigma^*)$ is a~centred
discrete $R_{q,0}$-valued Gaussian distribution of parameter $\sigma^*$. Since
$t$ would be in particular a~linear map, then for each $a^* (x)\in R_{q,0}$,
there would be exactly $|Ker(t)|$ many preimages of $a^* (x)$, hence the induced
map $t:R_q^2 \to R_{q,0}^2$ would take the uniform distribution over $R_q^2$ to
the uniform distribution over $R_{q,0}^2$. Likewise, and up to multiplication by
a scalar, it would take the PLWE distribution to the $R_{q,0}\times R_q$-PLWE
distribution, hence reducing the PLWE to the $R_{q,0}$-PLWE. However, it is not
clear at all for which rings (if any) such a~map $t$ exists.

\section{Coding examples}\label{s4}
To conclude our study, we provide numerical simulations of Algorithm~\ref{alg2}
for some specific sets of parameters. Our code has been developed with Maple 10
and is available at
\textsc{GitHub}\footnote{\url{https://github.com/raul-duran-diaz/PLWE-TraceAttack}}.
We have made no attempt at optimising our code, in particular, it does not
implement the automorphic evaluation of polynomials. This being said, we must
point out that the execution time is remarkably short, even in comparison with
the time necessary to obtain the sets of samples, for the parameters shown in
Table~\ref{tbl1}.

\subsection{Understanding our code}
Some remarks are in order to help understand our code: to begin with, we have not
simulated genuine discrete Gaussian distributions, which is a~non-trivial
problem but not entirely relevant when no high statistical accuracy is sought
after, as in most R/PLWE literature. Instead, we have discretised the regular
Gaussian distribution provided by Maple, using it as a~black box. As for uniform
distributions, we made use of Maple's random sampler \texttt{rand}, adjusted to produce
$\mathbb{F}_q$-samples.

Moreover, running each example, i.e., each Maple sheet, only requires choosing
the desired parameters in the ``main section'' and executing the sheet from the
beginning to the end. Notations for the Maple sheets follow closely those in
the present work in order to facilitate the reading and comprehension.

\subsection{Execution steps}
The execution consists of the following steps:
\begin{enumerate}
\item Initialising the uniform distribution (\texttt{rollq}) and the discrete
Gaussian (\texttt{X}).
\item Obtaining a~prime of the desired size meeting the hypotheses for
Theorem~\ref{chinese} to hold.
\item Obtaining the cyclotomic polynomial and its roots on an algebraic
extension, and assigning any one of them to the variable \texttt{rho}.
\item Obtaining the \emph{smallness region} $\Sigma$ for the input parameters.
\item Selecting a~number of executions (variable \texttt{ntests}) for
Algorithm~\ref{alg2}, and a~number of samples per execution (variable
\texttt{M}). Once these values have been assigned:
\begin{enumerate}
\item First, a~loop is executed \texttt{ntests} times and, for each turn,
\texttt{M}
samples from the PLWE oracle are generated, and passed to Algorithm~\ref{alg2}.
If it outputs anything different from a~set containing just one element, the
execution is counted as a~failure.
\item Second, another loop is executed, but this time producing the samples from
the
uniform oracle, and passing each set of samples to Algorithm~\ref{alg2}. If it
outputs anything different than an empty set, the execution is counted as a
failure.
\end{enumerate}
\end{enumerate}

\def\figurename{Table}

\subsection{Two execution examples}
Table~\ref{tbl1} shows the sets of parameters used for running two examples, and
Table~\ref{tbl2} presents a~secondary set of parameters, depending on
the selected ones in Table~\ref{tbl1}.

\begin{figure}[ht]
\centering
\begin{tabular}[c]
{lcc}
\hline
\textbf{Parameter} & \textbf{Example 1} & \textbf{Example 2} \\
\hline
$p$ & $2$ & $2$ \\
$n$ & $10$ & $11$ \\
$A$ & $2$ & $2$ \\
$q$ & $24029$ & $40013$ \\
$\sigma$ & $8$ & $8$ \\
\texttt{ntests} & $5$ & $5$ \\
\texttt{M} & $10$ & $10$ \\
\hline
\end{tabular}
\caption{Parameter selection for Examples 1 and 2}
\label{tbl1}
\end{figure}

\begin{figure}[ht]
\centering
\begin{tabular}[c]
{lcc}
\hline
\textbf{Dependent param} & \textbf{Example 1} & \textbf{Example 2} \\
\hline
Polynomial $\Phi$ & $x^{512}+1$ & $x^{1024}+1$ \\
$m$ & $1024$ & $2048$ \\
$N$ & $512$ & $1024$ \\
Factors of $\Phi$ over $\mathbb{F}_q$ & \small{$(x^{256}+11937)(x^{256}+12092)$}
                                      & \small{$(x^{512}+27481)(x^{512}+12532)$}
\\
$\rho$ & $-11937$ & $-27481$ \\
\hline
\end{tabular}
\caption{Dependent parameters for Examples 1 and 2}
\label{tbl2}
\end{figure}

\subsection{Conclusions on the numerical results}
To finish this section, several conclusions can be drawn.
\begin{itemize}
\item
First of all, it is interesting to remark that, regarding running times,
the most time consuming part is the process of sampling generation. For
Example 2, our hardware platform (Virtual Box configured with 1 GB of main
memory running over Intel CORE i5, @2.2 GHz) needs
about $200$ seconds to generate a~set of $10$ samples of any kind, but just
about 1 second for running Algorithm~\ref{alg2} over that set. In any case,
the attack is clearly feasible within very modest resource requirements.

\item
In the second place, but of much more interest, is the fact that execution
succeeds smoothly both for the PLWE and the uniform oracles, and this happens
for the two examples. This is in perfect agreement with the results predicted
by the theory, thus giving strong support to the effectiveness of the decision
attack presented in this work.
\end{itemize}

\subsubsection*{Acknowledgements}
I. Blanco-Chac\'on is partially supported by the Spanish National Research
Plan, grant no MTM2016-79400-P, by grant PID2019-104855RBI00, funded by MCIN / AEI / 10.13039 / 501100011033 and by the University of Alcal\'a grant CCG20/IA-057. 
R. Dur\'an-D\'iaz is partially supported by grant P2QProMeTe (PID2020-112586RB-I00), funded by MCIN / AEI / 10.13039 / 501100011033.
R.Y. Njah Nchiwo is supported by a PhD
scholarship from the Magnus Ehrnrooth Foundation, Finland, in part by Academy of Finland, grant 351271 (PI: C. Hollanti) and in part by MATINE Finnish Ministry of Defence, grant \#2500M-0147 (PI: C. Hollanti).
B. Barbero-Lucas is partially supported by the University of Alcal\'a grant CCG20/IA-057.


\EditInfo{April 4, 2023}{June 26, 2023}{Camilla Hollanti and Lenny Fukshansky}


{\small
\begin{thebibliography}{09}

\bibitem{ajtai}
M.~Ajtai.
\newblock Generating hard instances of lattice problems.
\newblock {\em Quaderni di Matematica}, 13:1--32, 2004.

\bibitem{blanco}
I.~Blanco-Chac\'on.
\newblock On the rlwe/plwe equivalence for cyclotomic number fields.
\newblock {\em Applicable Algebra in Engineering, Communication and Computing},
  33(1):53--71, 2022.

\bibitem{blanco1}
I.~Blanco-Chac\'on.
\newblock Rlwe/plwe equivalence for totally real cyclotomic subextensions via
  quasi-vandermonde matrices.
\newblock {\em Journal of Algebra and Its Applications}, 21:1--18, 2022.

\bibitem{blanco3}
I.~Blanco-Chac\'on, B.~Barbero-Lucas, R.~Dur\'an-D\'iaz, and R.~N. Nchiwo.
\newblock Cryptanalysis of plwe based on zero-trace quadratic roots
  (submitted).
\newblock {\em Quaderni di Matematica}, 13(1):1--32, 2023.

\bibitem{blanco2}
I.~Blanco-Chac\'on and L.~L\'opez-Hernanz.
\newblock Rlwe/plwe equivalence for the maximal totally real subextension of
  the $2^rpq$-th cyclotomic field.
\newblock {\em Advances in Mathematics of Communications}, 13(1):1--32, 2022.

\bibitem{dd}
L.~Ducas and A.~Durmus.
\newblock Ring-(lwe) in polynomial rings.
\newblock In {\em Public Key Cryptography -- PKC 2012}, pages 34--51. Springer
  Berlin Heidelberg, 2012.

\bibitem{ERS:2012:PEF}
M.~Elia, J.~Rosenthal, and D.~Schipani.
\newblock Polynomial evaluation over finite fields: new algorithms and
  complexity bounds.
\newblock {\em Applicable Algebra in Engineering, Communication and Computing},
  23(3):129--141, 2012.

\bibitem{ELOS:2015:PWI}
Y.~Elias, K.~Lauter, E.~Ozman, and K.~Stange.
\newblock Provably weak instances of ring-lwe.
\newblock In {\em Advances in Cryptology -- CRYPTO 2015}, pages 63--92. Lecture
  Notes in Computer Science. Springer Berlin Heidelberg, 2015.

\bibitem{ELOS:2016:RCN}
Y.~Elias, K.~Lauter, E.~Ozman, and K.~Stange.
\newblock Ring-lwe cryptography for the number theorist.
\newblock In {\em Directions in Number Theory}, pages 271--290. Springer
  International Publishing, 2016.

\bibitem{lpr}
V.~Lyubashevsky, C.~Peikert, and O.~Regev.
\newblock On ideal lattices and learning with errors over rings.
\newblock {\em Journal of the ACM}, 15:1--35, 2013.

\bibitem{micciancio}
D.~Micciancio.
\newblock The shortest vector in a lattice is hard to approximate to within
  some constant.
\newblock In {\em Proceedings of the 39-th Annual IEEE Symposium on Foundations
  of Computer Science}, pages 34--51. Springer Berlin Heidelberg, 1998.

\bibitem{regev}
O.~Regev.
\newblock On lattices, learning with errors, random linear codes, and
  cryptography.
\newblock In {\em S{TOC}'05: {P}roceedings of the 37th {A}nnual {ACM}
  {S}ymposium on {T}heory of {C}omputing}, pages 84--93. ACM, New York, 2005.

\bibitem{rsw}
M.~Rosca, D.~Stehl{\'e}, and A.~Wallet.
\newblock On the ring-lwe and polynomial-lwe problems.
\newblock In {\em Advances in Cryptology -- EUROCRYPT 2018}, pages 146--173.
  Springer International Publishing, 2018.

\bibitem{italians}
A.~J.~D. Scala, C.~Sanna, and E.~Signorini.
\newblock On the condition number of the vandermonde matrix of the nth
  cyclotomic polynomial.
\newblock {\em Journal of Mathematical Cryptology}, 15(1):174--178, 2021.

\bibitem{italians2}
A.~J.~D. Scala, C.~Sanna, and E.~Signorini.
\newblock Rlwe and plwe over cyclotomic number fields are not equivalent.
\newblock {\em Applicable Algebra in Engineering, Communication and Computing},
  22:174--178, 2022.

\bibitem{stehle}
D.~N. Stehle, R.~Steinfeld, K.~Tanaka, and K.~Xagawa.
\newblock Efficient public key encryption based on ideal lattices.
\newblock In {\em Advances in Cryptology ASIACRYPT 2009}, pages 617--635.
  Springer Berlin Heidelberg, 2009.

\bibitem{WZFY:2017:EFC}
H.~Wu, L.~Zhu, R.~Feng, and S.~Yang.
\newblock Explicit factorizations of cyclotomic polynomials over finite fields.
\newblock {\em Designs, Codes and Cryptography}, 83(1):197--217, 2017.

\end{thebibliography}
}
\end{document}